\documentclass[journal,onecolumn,12pt]{IEEEtran}
\usepackage{graphicx,psfrag,subfigure,url,amsmath,amsfonts,
epsfig,amssymb}

\newtheorem{claim}{Claim}

\newtheorem{remark}{Remark}
\newtheorem{bound}{Bound}

\newcommand{\p}{{\rm P}}
\long\def\symbolfootnote[#1]#2{\begingroup
\def\thefootnote{\fnsymbol{footnote}}\footnote[#1]{#2}\endgroup}

\begin{document}

\title{A note on  outer bounds for broadcast channel}
\author{Chandra Nair\thanks{This result was presented as a part of the talk by the author in the International Zurich Seminar 2010. The purpose of this note is to serve as a documented proof of this fact.}}

\maketitle

\begin{abstract}
In this note we establish two facts concerning the so-called {\em New-Jersey} outer bound. We show that this outer bound is equivalent to a much simpler  {\em computable} region; and secondly we show that in the absence of private information this bound is exactly same as the $UV$-outerbound.
\end{abstract}
\section{Introduction}
Recently there has a flurry of activity on the outer bounds for the two-receiver broadcast channel. This work is an attempt to clean up the bounds and present a simple and clean picture. In this note we study the {\em New-Jersey} outer bound for
2-receiver discrete memoryless broadcast channel.   We show that this bound is equivalent to another bound (Bound \ref{bd:outer}). Further, the equivalent bound presented is also computable, i.e. the auxiliary random variables have bounded cardinalities, where as the original bound is not computable.

\begin{bound}
 \label{bd:outer3}
[New-Jersey region: \cite{lks08}] The closure of the union of rate
triples $(R_0,R_1,R_2)$ satisfying
{\small \begin{align*}
 R_0 &\leq \min \{I(T;Y|W_1), I(T;Z|W_2)\}\\
 R_1 &\leq I(U;Y|W_1)\\
 R_2 &\leq I(V;Z|W_2)\\
 R_0+R_1&\leq I(T,U;Y|W_1)\\
 R_0+R_1&\leq I(U;Y|T,W_1,W_2)+ I(T,W_1;Z|W_2)\\
 R_0+R_2&\leq I(T,V;Z|W_2)\\
 R_0+R_2&\leq I(V;Z|T,W_1,W_2)+ I(T,W_2;Y|W_1)\\
 R_0+R_1+R_2&\leq I(U;Y|T,V,W_1,W_2)+ I(T,V,W_1;Z|W_2)\\
 R_0+R_1+R_2&\leq I(V;Z|T,U,W_1,W_2)+ I(T,U,W_2;Y|W_1)\\
 R_0+R_1+R_2&\leq I(U;Y|T,V,W_1,W_2)+ I(T,W_2;Y|W_1)+ I(V;Z|T,W_1,W_2)\\
 R_0+R_1+R_2&\leq I(V;Z|T,U,W_1,W_2)+ I(T,W_1;Z|W_2)+ I(U;Y|T,W_1,W_2)
\end{align*}}
for some $p(u)p(v)p(t)p(w_1,w_2|u,v,t)p(x|u,v,t,w_1,w_2)p(y,z|x)$
constitutes an outer bound. Further one can restrict $X$ to be a
deterministic function of $(u,v,t,w_1,w_2)$ and $U,V,T$ are
uniformly distributed.
\end{bound}

\begin{remark}
The follow two points are worth noting:
\begin{itemize}
\item In \cite{lks08}, the authors note that Bound \ref{bd:outer3} is at least as good as the previously known bounds; however it is not clear if it is any better. It turns out that when $R_0=0$, this note will prove that this is no better than a previously known bound.
\item The bound presented above differs slightly from the {\em New-Jersey} region\cite{lks08}. This is based on comments made to the author by Amin Gohari (based on his joint observation with Venkat Anantharam). Firstly, the region is the closure of the rate pairs (apriori, in the absence of cardinality bounds, it is not clear that the union is closed), and secondly the terms $I(T,W_2;Y|W_1)$, $I(T,W_1;Z|W_2)$ replace the terms $ I(T,W_1,W_2;Y)$ and $I(T,W_1,W_2;Z)$ respectively in the last two inequalities. The same argument in the New Jersey bound shows that this expression is also a valid outer bound (and at least as good as the New Jersey bound). 
\end{itemize}
\end{remark}
\medskip

Now consider the following much simpler region.

\medskip
\begin{bound}
 \label{bd:outer2}
The union of rate triples $(R_0,R_1,R_2)$ that satisfy the following
inequalities
{\small\begin{align*}
 R_0 &\leq \min \{I(W;Y), I(W;Z)\}\\
 R_0+R_1 &\leq I(U;Y|W) + \min \{I(W;Y), I(W;Z)\}\\\
 R_0 + R_2 &\leq I(V;Z|W)+ \min \{I(W;Y), I(W;Z)\}\\ 
 R_0+R_1+ R_2 &\leq \min \{I(W;Y), I(W;Z)\} + I(U;Y|V,W)+ I(V;Z|W)\\
 R_0+R_1+ R_2 &\leq \min \{I(W;Y), I(W;Z)\} + I(U;Y|W)+ I(V;Z|U,W)
\end{align*}}
over all $p(u,v,w,x)$ such that $(U,V,W) \to X \to (Y,Z)$
forms a Markov chain.
\end{bound}
\begin{remark}
We make the following simple observations:
\begin{itemize}
\item Setting $W' = (T,W_1,W_2), U' = U$ and $V' =V$ into Bound \ref{bd:outer2} we find that
Bound \ref{bd:outer3} $\subseteq$ Bound \ref{bd:outer2}.
\item Bound \ref{bd:outer2} can be directly obtained as an outer bound following standard manipulations.
\end{itemize}
\end{remark}

In the next section we will evaluate Bound \ref{bd:outer3} and show
the other non-trivial direction that Bound \ref{bd:outer2}
$\subseteq$ Bound \ref{bd:outer3}.

\section{The equivalence between Bound \ref{bd:outer3} and Bound \ref{bd:outer2}}
\label{se:genrem}

Clearly from the remark it suffices to show that Bound \ref{bd:outer2}
$\subseteq$ Bound \ref{bd:outer3}. The main idea is borrowed from a trick in \cite{naw08}.
We show that the constraints on the auxiliary random variables described by the following:
\begin{itemize}
\item the union over distributions $p(u)p(v)p(t)p(w_1,w_2|u,v,t)p(x|u,v,t,w_1,w_2)p(y,z|x)$
\item Restrict $X$ to be a
deterministic function of $(u,v,t,w_1,w_2)$ and $U,V,T$ are
uniformly distributed
\end{itemize}
 are  in effect  {\em red herrings}, in the sense that even if we take the union over
all $p(u,v,t, w_1, w_2,x)p(y,z|x)$, we arrive at the same region. That is there is no real advantage 
in making our distributions more restrictive in the above sense; further, these restrictions prevent us from 
establishing cardinality bounds on the region.

Denote the region obtained by Bound \ref{bd:outer3} by $\mathcal{R}$ and let
$\mathcal{R}_1$ be the region obtained by the same constraints  but with the union taken over 
$p(u,v,t, w_1, w_2,x)$, such that $X$ is a function of $(U,V,T,W_1,W_2)$.
 We will show these two regions are identical.

\medskip
\begin{claim}
$\mathcal{R}=\mathcal{R}_1$.
\end{claim}
\medskip

\begin{proof}
Clearly, $\mathcal{R} \subseteq \mathcal{R}_1$. Therefore it
suffices to show the non-trivial direction.

Given a $(U,V,T,W_1,W_2)$,  let $\mathcal{U} = \{0,1,...,m_1-1\},$ $\mathcal{V} =
\{0,1,...,m_2-1\},$ and $\mathcal{T} =
\{0,1,...,m_3-1\}$. Define new random variables
$U^*,V^*,T^*,W_1^*,W_2^*$ and a distribution
$p(u^*,v^*,t^*,w_1^*,w_2^*,x)$ according to
 {\small \begin{align}
  &\p(U^*=u, V^* = v, T^* = t,W_1^* = (w_1,i,k), W_2^* = (w_2,j,k), X^*=x) \nonumber \\
   & \quad= \frac{1}{m_1m_2m_3} \p(U=(u+i)_{m_1}, V = (v+j)_{m_2}, T = (t+k)_{m_3},W_1 = w_1, W_2 = w_2, X=x) \label{eq:eq0}
 \end{align}}
where $(\cdot)_{m_i}$ denotes the $\mod$ operation, and $i,j,k$ takes values in $[0:m_1-1], [0:m_2-1], [0:m_3-1]$ respectively.
Note that if $X$ is a function of $(U,V,T,W_1,W_2)$, then  $X^*$ is a function of $(U^*,V^*,T^*,W_1^*,W_2^*)$.

It is straightforward to check the following:
{\small \begin{align}
 \p(U^*=u,V^*=v,T^*=t) &=\frac{1}{m_1m_2m_3}  \label{eq:eq1}\\
 &\quad \mbox{ hence independent and uniformly distributed.} \nonumber \\
 \p(T^*=t, W_1^*=(w_1,i,k), X=x) &= \frac{1}{m_1 m_3} \p(T=(t+k)_{m_3},W_1=w_1,X=x)\label{eq:eq2}\\
 \p(T^*=t, W_2^*=(w_2,j,k), X=x) &= \frac{1}{m_2 m_3} \p(T=(t+k)_{m_3},W_2=w_2,X=x)\label{eq:eq3}\\
 \p(U^*=u,W_1^*=(w_1,i,k),X=x) &= \frac{1}{m_1m_3} \p(U=(u+i)_{m_1}, W_1=w_1,X=x)\label{eq:eq4}\\
 \p(V^*=v,W_2^*=(w_2,j,k),X=x) &= \frac{1}{m_2 m_3} \p(V=(v+j)_{m_2},
 W_2=w_2,X=x)\label{eq:eq5}\\
 \p(U^*=u,T^*=t, W_1^*=(w_1,i,k),X=x) &= \frac{1}{m_1m_3} \p(U=(u+i)_{m_1}, T=(t+k)_{m_3},W_1=w_1,X=x)\label{eq:eq6}\\
 \p(V^*=v,T^*=t, W_2^*=(w_2,j,k),X=x) &= \frac{1}{m_2m_3} \p(V=(v+j)_{m_2}, T=(t+k)_{m_3},W_2=w_2,X=x)\label{eq:eq7}.
\end{align}}

Similarly, one also obtains
{\small \begin{align}
 & \p(T^*=t, W_1^*=(w_1,i,k),W_2^*=(w_2,j,k),X=x) \nonumber\\
&\qquad = \frac{1}{m_1m_2m_3} \p(T=(t+k)_{m_3},W_1=w_1,W_2=w_2,X=x)\label{eq:eq8}\\
& \p(U^*=u, T^*=t, W_1^*=(w_1,i,k),W_2^*=(w_2,j,k),X=x) \nonumber \\
 & \qquad = \frac{1}{m_1m_2m_3} \p(U=(u+i)_{m_1},T=(t+k)_{m_3},W_1=w_1,W_2=w_2,X=x)\label{eq:eq9} \\
 &\p(V^*=v, T^*=t, W_1^*=(w_1,i,k),W_2^*=(w_2,j,k),X=x) \nonumber \\
 & \qquad = \frac{1}{m_1m_2m_3} \p(V=(v+j)_{m_2},T=(t+k)_{m_3},W_1=w_1,W_2=w_2,X=x) \label{eq:eq10}
\end{align}}

From the above it follows a straightforward manner that the following equalities hold:
{\small \begin{align*}
I(T^*;Y^*|W_1^*)&= I(T;Y|W_1) ~ (\mbox{from}~\eqref{eq:eq2})\\
I(T^*;Z^*|W_2^*)&= I(T;Z|W_2) ~ (\mbox{from}~\eqref{eq:eq3})\\
I(U^*;Y^*|W_1^*)&=I(U;Y|W_1) ~ (\mbox{from}~\eqref{eq:eq4})\\
I(V^*;Z^*|W_2^*)&=I(V;Z|W_2) ~ (\mbox{from}~\eqref{eq:eq5})\\
I(T^*, U^*;Y^*|W_1^*)&=I(T, U;Y|W_1)  ~ (\mbox{from}~\eqref{eq:eq6})\\
I(T^*, V^*;Z^*|W_2^*)&=I(T, V;Z|W_2)  ~ (\mbox{from}~\eqref{eq:eq7})\\
I(T^*,W_1^*,W_2^*;Y^*)&=I(T,W_1,W_2;Y)   ~ (\mbox{from}~\eqref{eq:eq8})\\
I(T^*,W_1^*,W_2^*;Z^*)&=I(T,W_1,W_2;Z)   ~ (\mbox{from}~\eqref{eq:eq8})\\
I(T^*,W_2^*;Y^*|W_1^*)&=I(T,W_2;Y|W_1) ~ (\mbox{from}~\eqref{eq:eq8}) \\
I(T^*,W_1^*;Z^*|W_2^*)&=I(T,W_1;Z|W_2) ~ (\mbox{from}~\eqref{eq:eq8}) \\
I(U^*;Y^*|T^*,W_1^*,W_2^*)&=I(U;Y|T,W_1,W_2) ~ (\mbox{from}~\eqref{eq:eq9})\\
I(V^*;Z^*|T^*,W_1^*,W_2^*)&=I(V;Z|T,W_1,W_2) ~ (\mbox{from}~\eqref{eq:eq10})\\
I(U^*;Y^*|T^*,V^*,W_1^*,W_2^*)&=I(U;Y|T,V,W_1,W_2) ~ (\mbox{from}~\eqref{eq:eq0}, \eqref{eq:eq10})\\
I(V^*;Z^*|T^*,U^*,W_1^*,W_2^*)&=I(V;Z|T,U,W_1,W_2)  ~ (\mbox{from}~\eqref{eq:eq0}, \eqref{eq:eq9})
\end{align*}}

Hence all the terms that appears in bound are preserved, and hence $\mathcal{R}_1 \subseteq \mathcal{R}$. Hence the region
described by Bound \ref{bd:outer3} is indeed the same as that described by the same expression, if we remove the structure on $(U,V,T,W_1,W_2)$..
\end{proof}
\medskip

Given $(U,V,W,X)$ such that $X$ is a function of $U,V,W$, 
set $U' = (U,W), V' = (V,W), T' =W, W_1' = \emptyset$ and $W_2' =
\emptyset$. Plugging this choice into the constraints we see that the following region is a subset of $\mathcal{R}_1$,
{\small \begin{align}
 R_0 &\leq \min \{I(W;Y), I(W;Z)\} \nonumber\\
 R_1 &\leq I(U,W;Y) \label{eq:redu1}\\
 R_2 &\leq I(V,W;Z) \label{eq:redu2}\\
 R_0+R_1 &\leq I(U,W;Y) \nonumber\\
 R_0+R_1 &\leq I(U;Y|W) + I(W;Z) \nonumber\\
 R_0 + R_2 &\leq I(V,W;Z) \nonumber \\
 R_0 + R_2 &\leq I(V;Z|W)+ I(W;Y) \nonumber \\
 R_0+R_1+ R_2 &\leq I(U,W;Y)+ I(V;Z|U,W)\nonumber \\
 R_0+R_1+ R_2 &\leq I(V,W;Z)+ I(U;Y|V,W)\nonumber \\
 R_0+R_1+ R_2 &\leq I(W;Y)+ I(U;Y|V,W)+ I(V;Z|W)\nonumber \\
 R_0+R_1+ R_2 &\leq I(W;Z)+ I(U;Y|W)+ I(V;Z|U,W) \nonumber
\end{align}}
Obviously the inequalities \ref{eq:redu1} and \ref{eq:redu2} are
redundant (due to non-negativity of rates. Thus the region described by Bound \ref{bd:outer2} (with the additional constraint that
$X$ is a function of $U,V,W$) is a subset of  $\mathcal{R}_1$. 

\begin{remark}
However it is easy to see\footnote{standard argument: set $U'=U,Q$ where $Q$ is independent of $(U,V,W)$ such that $X$ is a function of $U,V,W,Q$ always increases(may not be strict) the region.} that the constraint 
$X$ is a function of $U,V,W$ still yields the entire region described by Bound \ref{bd:outer2}.
\end{remark}

Thus we have completed the proof of the other direction, i.e. we have shown that Bound \ref{bd:outer2} $\subseteq$
$\mathcal{R}_1= \mathcal{R}$, that given by Bound \ref{bd:outer3}). This shows that the regions given by the two bounds, Bound \ref{bd:outer3} and Bound \ref{bd:outer2}, are equal.

\section{UVW-Outer bound}

From the previous section, it is clear that the following computable region is equivalent to the New-Jersey outer bound.

\begin{bound}
 \label{bd:outer} [The UVW outer bound]
The union of rate triples $(R_0,R_1,R_2)$ that satisfy the following
inequalities
\begin{align*}
 R_0 &\leq \min \{I(W;Y), I(W;Z)\}\\
 R_0+R_1 &\leq I(U;Y|W) + \min \{I(W;Y), I(W;Z)\}\\\
 R_0 + R_2 &\leq I(V;Z|W)+ \min \{I(W;Y), I(W;Z)\}\\ 
 R_0+R_1+ R_2 &\leq \min \{I(W;Y), I(W;Z)\} + I(X;Y|V,W)+ I(V;Z|W)\\
 R_0+R_1+ R_2 &\leq \min \{I(W;Y), I(W;Z)\} + I(U;Y|W)+ I(X;Z|U,W)
\end{align*}
over all $p(u,v,w)p(x|u,v,w)$ such that $(U,V,W) \to X \to (Y,Z)$
forms a Markov chain. Further  it suffices to consider $|W| \leq |X| + 5, |U| \leq |X| + 1, |W| \leq |X| + 1$.
\end{bound}

\begin{remark}
The replacement of $I(U;Y|V,W)$ with $I(X;Y|V,W)$ and $I(V;Z|U,W)$ with
$I(X;Z|U,W)$ follows from the fact that we can assume that $X$ is a function of $U,V,W$ (w.l.o.g.). 
The proof of the cardinality bounds are established below.
\end{remark}

\subsection{Computability of the outer bound}

\begin{claim}
To compute the region in Bound \ref{bd:outer} it suffices to consider $|W| \leq |X| + 5, |U| \leq |X| + 1, |W| \leq |X| + 1$.
\end{claim}
\begin{proof}
The proof of this claim follows from standard arguments: i.e. the Fenchel-Bunt extension to Caratheodory's theorem. We require $W'$ to preserve the distribution of $X$, and the values of $H(Y|W), H(Z|W)$, $H(Y|U,W), H(Z|U,W)$, $H(Y|V,W), H(Z|V,W)$. Hence we can choose a $W$ of cardinality at most  $|X| + 5$, and distributions $p(u,v,x|w)$ so as to preserve these quantities. (Observe that to preserve the distribution of $X$ we only need $|X| - 1$ constraints.) Conditioned $W=w$, we can choose $U$ of size $|X|+1$ so that $p(X|W=w), H(Y|U,W=w), H(Z|U,W=w)$ is preserved. Similarly, we can choose a $V$ of size $|X|+1$ so that $p(X|W=w), H(Y|V,W=w), H(Z|V,W=w)$ is preserved. Hence to compute the region, we can make the restrictions as stated in the claim. This makes the outer bound (Bound \ref{bd:outer}) computable.
\end{proof}

\subsection{Private messages outer bound}

In this section we study the outer bound for the case $R_0=0$ (there is no common message).

\begin{claim}
Whe $R_0=0$, then Bound 3 is equivalent to an earlier bound (UV-outerbound \cite{nae07}), which states that the union of rate triples $(R_1,R_2)$ that satisfy the following
inequalities
\begin{align*}
 R_1 &\leq I(U;Y) \\
R_2 &\leq I(V;Z)\\ 
R_1+ R_2 &\leq  I(X;Y|V)+ I(V;Z)\\
R_1+ R_2 &\leq  I(U;Y)+ I(X;Z|U)
\end{align*}
over all $p(u,v)p(x|u,v)$ such that $(U,V) \to X \to (Y,Z)$
forms a Markov chain form an outer bound to the capacity region.
\end{claim}
\begin{proof}
Given a $(U,V,X)$ and the corresponding region in the UV-outer bound, by setting $W =\emptyset, U=U, V=V$ we see that this region is contained in the region described by Bound \ref{bd:outer}. Hence UV-outer bound $\subseteq$ Bound \ref{bd:outer}.

Given a $(U,V,W,X)$ and the corresponding region in Bound \ref{bd:outer}, by setting
$U = (U,W); V = (V,W)$ we see that this region (in Bound \ref{bd:outer}) is contained in the region described by the UV-outer bound. Hence UV-outer bound $\supseteq$ Bound \ref{bd:outer}.

Hence the two regions are equivalent.
\end{proof}

\section{Conclusion}

In this work we summarize the state of the various recent attempts at writing an outer bound for the two-receiver broadcast channel. It is not clear whether the UVW outer bound in Bound \ref{bd:outer} is any better than the one presented in \cite{nae07} (it is at least as good). The main contribution here is to collect all the developments on the outer bound, and present a simple bound that is as good as all the outer bounds currently developed. More importantly, we also make the outer bound region computable.

Hopefully this note can be useful to the researchers and to the community in general while trying to figure out the relationship between the bounds.

\section*{Acknowledgements}

The author wishes to thank Vincent Wang, Amin Gohari, and Abbas El Gamal; all of whom played a direct or indirect role in the preparation of this note.

\bibliographystyle{amsplain}

\end{document}